\documentclass[11pt]{article}
\usepackage{fullpage} 
\usepackage[indent]{parskip} 

\usepackage[T1]{fontenc} 
\usepackage{lmodern} 
\usepackage{newtxtext}
\makeatletter%
\@ifpackageloaded{xeCJK}{}{%
    \usepackage[utf8]{inputenc}}
\makeatother

\usepackage[shortlabels,inline]{enumitem}

\usepackage{amsmath,amssymb,amsfonts}
\usepackage{dsfont}
\usepackage{bm}
\usepackage[mathscr]{eucal} 
\usepackage{amsthm}
\usepackage[mleftright]{diffcoeff}
	\difdef{f, s, fp, c}{}{%
		outer-Ldelim=\left.,%
		outer-Rdelim=\right|,%
		sub-nudge=0mu}%
	\difdef{c}{J}{%
		op-symbol=\rmJ,%
		op-sub-nudge=-2mu}%
\usepackage[group-separator={,}]{siunitx} 

\usepackage{xcolor}
\usepackage[pdftex]{graphicx} 
\usepackage{pgfplots}
	\pgfplotsset{compat=1.18}
	\usetikzlibrary{arrows}
\usepackage[export]{adjustbox}
\makeatletter%
\@ifclassloaded{ieeecolor}
	{\usepackage[caption=false,font=footnotesize]{subfig}}
	{\usepackage{subcaption}
	\captionsetup[figure]{font=footnotesize}}
\makeatother

\makeatletter%
\@ifpackageloaded{natbib}{}{%
    \usepackage[noadjust]{cite}} 
\makeatother

\usepackage{lipsum}
\usepackage{nameref}
\usepackage{silence}
\usepackage{verbatim} 
\usepackage{url} 
\usepackage{endnotes}

\makeatletter
\newcommand{\disablepackage}[2]{%
  \disable@package@load{#1}{#2}}
\makeatother

    \setlist{nosep,leftmargin=*}

\usepackage[colorlinks,citecolor=teal]{hyperref}

\theoremstyle{plain}
\newtheorem{thm}{Theorem}

\newtheorem{lem}{Lemma}

\theoremstyle{definition}

\newtheorem{ass}{Assumption}

\theoremstyle{remark}

\newtheorem*{clm*}{Claim}

\makeatletter%
\@tempcnta=\@ne
\def\@nameedef#1{\expandafter\edef\csname #1\endcsname}
\loop\ifnum\@tempcnta<27
    \@nameedef{bb\@Alph\@tempcnta}{{\noexpand\mathbb{\@Alph\@tempcnta}}}
    \@nameedef{bf\@Alph\@tempcnta}{{\noexpand\mathbf{\@Alph\@tempcnta}}}
    \@nameedef{bm\@Alph\@tempcnta}{{\noexpand\bm{\@Alph\@tempcnta}}}
    \@nameedef{cal\@Alph\@tempcnta}{{\noexpand\mathcal{\@Alph\@tempcnta}}}
    \@nameedef{scr\@Alph\@tempcnta}{{\noexpand\mathscr{\@Alph\@tempcnta}}}
    \@nameedef{ds\@Alph\@tempcnta}{{\noexpand\mathds{\@Alph\@tempcnta}}}
    \@nameedef{frak\@Alph\@tempcnta}{{\noexpand\mathfrak{\@Alph\@tempcnta}}}
    \@nameedef{rm\@Alph\@tempcnta}{{\noexpand\mathrm{\@Alph\@tempcnta}}}
    \@nameedef{sf\@Alph\@tempcnta}{{\noexpand\mathsf{\@Alph\@tempcnta}}}
    \@nameedef{op\@Alph\@tempcnta}{{\noexpand\operatorname{\@Alph\@tempcnta}}}
    \@nameedef{hat\@Alph\@tempcnta}{{\noexpand\hat{\@Alph\@tempcnta}}}
    \@nameedef{tilde\@Alph\@tempcnta}{{\noexpand\tilde{\@Alph\@tempcnta}}}
    \@nameedef{ul\@Alph\@tempcnta}{{\noexpand\underline{\@Alph\@tempcnta}}}
    \@nameedef{ol\@Alph\@tempcnta}{{\noexpand\overline{\@Alph\@tempcnta}}}
	\@nameedef{v\@Alph\@tempcnta}{{\noexpand\vec{\@Alph\@tempcnta}}}
    \advance\@tempcnta\@ne
\repeat
\makeatother

\let\vec\undefined
\newcommand*{\vec}[1]{\bm{#1}}

\newcommand*{\Z}{\mathbb{Z}}

\newcommand*{\R}{\mathbb{R}}

\newcommand*{\K}{\mathcal{K}}
\newcommand*{\KInf}{\K_\infty}

\let\ForAll\forall
\let\forall\undefined
\DeclareMathOperator\forall{\ForAll}
\let\Exists\exists
\let\exists\undefined
\DeclareMathOperator\exists{\Exists}

\DeclareMathOperator*{\esssup}{ess\,sup}

\DeclareMathOperator{\eigMax}{\overline\lambda}
\DeclareMathOperator{\eigMin}{\underline\lambda}

\title{\LARGE\bf Input-to-state stabilization of linear systems under data-rate constraints}
\author{Mahmoud Zamani and Guosong Yang%
    \thanks{The authors are with the Department of Electrical and Computer Engineering, Rutgers University--New Brunswick, Piscataway, NJ 08854 USA (e-mails: {\tt \{mahmoud.zamani,\,guosong.yang\}@rutgers.edu}).}%
}
\date{}
    
\begin{document}
\maketitle

\begin{abstract}
We study feedback stabilization of linear systems under data-rate constraints in the presence of completely unknown disturbances.
A communication and control strategy is proposed based on sampled and quantized state measurements, where the quantization range is dynamically adjusted using reachable-set approximations and a disturbance estimate derived from quantization parameters.
The strategy alternates between stabilizing and searching stages to recapture the state after escapes from the quantization range.
Under a data-rate condition, it guarantees input-to-state stability (ISS) with respect to the disturbance.
An additional quantization symbol is introduced to establish ISS near the equilibrium.
A simulation example illustrates the effectiveness of the proposed approach.
\end{abstract}

\section{Introduction}\label{sec:intro}
Feedback control under data-rate constraints has been an active research area for decades, as surveyed in, e.g., \cite{NairFagnaniZampieriEvans2007,JiangLiu2013,ParkErgenFischioneLuJohansson2018}.
Such constraints arise naturally in networked control systems due to communication costs, bandwidth limitations, and security considerations.
Beyond these practical motivations, a fundamental question is how much information is required to achieve a given control objective.

A widely used framework for achieving finite data rate is to generate the control input from sampled and quantized state measurements taking values in a finite set.
Beginning with \cite{Delchamps1990,EliaMitter2001,BrockettLiberzon2000,Liberzon2003,HespanhaOrtegaVasudevan2002,Liberzon2003TAC,TatikondaMitter2004}, various quantization schemes have been developed for stabilizing linear systems.
In particular, asymptotic stabilization under finite data rate requires dynamic adjustment of the size of the quantization range (zooming) \cite{BrockettLiberzon2000,HespanhaOrtegaVasudevan2002,Liberzon2003TAC,Liberzon2003,TatikondaMitter2004}, while the required data rate can be reduced by dynamically adjusting its center as well (moving-center quantization) \cite{HespanhaOrtegaVasudevan2002,Liberzon2003TAC,TatikondaMitter2004}.
These ideas have subsequently been extended to nonlinear systems \cite{Liberzon2003,LiberzonHespanha2005,Ferdinando2022} and switched systems \cite{Liberzon2014,YangLiberzon2018}. 

We consider feedback stabilization under data-rate constraints in the presence of unknown disturbances.
In \cite{HespanhaOrtegaVasudevan2002,TatikondaMitter2004}, a known bound on the disturbance is assumed, and the minimum data rate required for asymptotic stabilization is characterized.
Without such a bound, the problem becomes significantly more challenging, as disturbances may drive the state outside the quantization range after capture.
In this setting, \cite{LiberzonNesic2007} established input-to-state stability (ISS) \cite{Sontag1989} using dynamic quantization with alternating zooming-out and zooming-in stages.
Similar ISS results were obtained in \cite{Wang2017} using logarithmic quantization in polar coordinates, and in \cite{SharonLiberzon2012} using moving-center quantization with improved data-rate efficiency.
However, the fixed-center quantization schemes in \cite{LiberzonNesic2007,Wang2017} lead to implicit and generally conservative data-rate bounds, whereas the moving-center scheme in \cite{SharonLiberzon2012} relies on additional mechanisms such as a dedicated escape-detection mode and quantizer resets.
More recently, \cite{YangLiberzon2018} proposed a disturbance-estimation framework for switched linear systems that avoids these complexities, but achieves only the weaker property of practical ISS \cite{SontagWang1996}. 

This paper addresses these limitations for linear systems with completely unknown disturbances.
Extending the disturbance-estimation idea of \cite{YangLiberzon2018}, we propose a communication and control strategy that combines the data-rate efficiency of moving-center quantization with a Lyapunov-based design and eliminates the need for a dedicated escape-detection mode or quantizer resets.
Our results provide an explicit characterization of the admissible data rate and establish ISS rather than practical ISS.
To this end, we design a disturbance estimate derived from quantization parameters to replace the estimate based on an increasing sequence of constants in \cite{YangLiberzon2018}, and introduce an additional quantization symbol for establishing ISS near the equilibrium.
These features distinguish the proposed approach from \cite{LiberzonNesic2007,SharonLiberzon2012,Wang2017,YangLiberzon2018} and constitute the main methodological contribution of the paper.

The remainder of the paper is organized as follows.
Section~\ref{sec:pre} introduces the system model and information structure.
Section~\ref{sec:main} presents the main theorem.
Section~\ref{sec:strat} describes the communication and control strategy, while Section~\ref{sec:reach} derives the update formulas using reachable-set approximations.
Section~\ref{sec:proof} provides the stability analysis, with key intermediate results summarized as technical lemmas.
Section~\ref{sec:sim} presents a simulation example.
Section~\ref{sec:end} concludes the paper with a summary and an outlook for future research.

\emph{Notation:}
The $ \infty $-norm of a vector $ v = (v_1, \ldots, v_n) \in \R^n $ is denoted by $ |v| := |v|_\infty = \max_{1 \leq i \leq n} |v_i| $.
The induced $ \infty $-norm of a matrix $ M = [a_{ij}] \in \R^{n \times n} $ is denoted by $ \|M\| := \|M\|_\infty = \max_{1 \leq i \leq n} \sum_{j=1}^{n} |a_{ij}| $.
For a symmetric matrix $ M \in \R^{n \times n} $, the largest and smallest eigenvalues are denoted by $ \eigMax(M) $ and $ \eigMin(M) $, respectively. 
A continuous function $ \gamma: \R_{\geq 0} \to \R_{\geq 0} $ is of class $ \KInf $, denoted by $ \gamma \in \KInf $, if it is strictly increasing, unbounded, and satisfies $ \gamma(0) = 0 $.
The left limit of a piecewise continuous function $ z(\cdot) $ at $ t $ is denoted by $ z(t^-) := \lim_{s \nearrow t} z(s) $.

\section{Problem formulation}\label{sec:pre}

\subsection{System definition}\label{ssec:pre-sys}
Consider a continuous-time linear system
\begin{equation}\label{eq:lin}
    \dot x = A x + B u + D d, \qquad x(0) = x_0,
\end{equation}
where $ x \in \R^{n_x} $ is the state, $ u \in \R^{n_u} $ is the control input, and $ d \in \R^{n_d} $ is an unknown disturbance.
The disturbance $ d(\cdot) $ is assumed to be Lebesgue measurable and locally essentially bounded.
The essential supremum $ \infty $-norm of $ d(\cdot) $ over an interval $ J $ is denoted by $ \|d\|_{J} := \esssup_{t \in J} |d(t)| $. 

Our first basic assumption concerns stabilizability of \eqref{eq:lin}.
\begin{ass}[Stabilizability]\label{ass:stabilizable}
The pair $ (A, B) $ is stabilizable, that is, there exists a state feedback gain matrix $ K $ such that $ A + B K $ is Hurwitz (all eigenvalues have negative real parts).
\end{ass}
In what follows, we assume that such a matrix $ K $ has been selected and fixed.

\subsection{Information structure}\label{ssec:pre-info}
We aim to generate a stabilizing control input $ u(\cdot) $ based on limited information about the state $ x(\cdot) $.

\begin{figure}[!t]
\centering
\includegraphics[width=.4\textwidth]{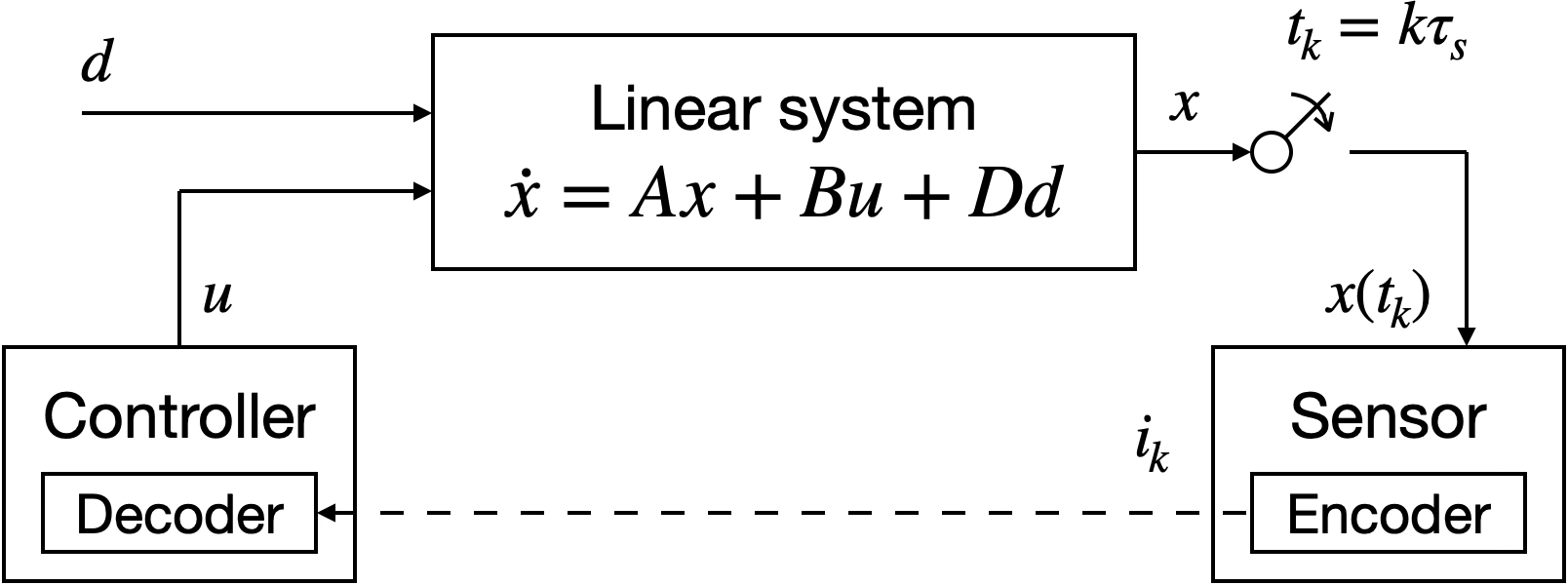}
\caption{Information structure.}
\label{fig:sys-info}
\end{figure}
As shown in Fig.~\ref{fig:sys-info}, the feedback loop includes a sensor with an encoder and a controller with a decoder.
The sensor samples the state at times $ t_k = k \tau_s,\, k \in \Z_{\geq 0} $, where $ \tau_s > 0 $ is a fixed \emph{sampling period}.
Each sample $ x(t_k) $ is encoded as an integer $ i_k \in \{0, 1, \ldots, N^{n_x} + 1\} $, where $ N > 1 $ is a fixed integer, and then transmitted to the controller.
The resulting data transmission rate is
\[
	\frac{\log_2(N^{n_x} + 2)}{\tau_s}
\]
bits per unit time. 
This information structure enables a separation of sensing and control tasks, since the controller does not require access to the exact state.
Details of the communication and control strategy are presented in Section~\ref{sec:strat}.

Our second basic assumption concerns the data rate.
\begin{ass}[Data rate]\label{ass:data-rate}
The sampling period $ \tau_s $ and the integer $ N $ satisfy
\begin{equation}\label{eq:data-rate-bnd}
    \Lambda := \|e^{A \tau_s}\| < N.
\end{equation}
\end{ass}
The inequality in \eqref{eq:data-rate-bnd} requires $ N $ to be sufficiently large relative to $ \tau_s $, thereby imposing a lower bound on the admissible data rate.
Similar data-rate conditions were shown to be sufficient for stabilizing linear systems \cite{Liberzon2003TAC,SharonLiberzon2012} and switched linear systems \cite{Liberzon2014,YangLiberzon2018}.
\footnote{In contrast to \cite{Liberzon2003TAC,SharonLiberzon2012,Liberzon2014,YangLiberzon2018}, the proposed approach does not require the integer $ N $ to be odd, owing to an additional quantization symbol introduced for establishing ISS near the equilibrium; see Section~\ref{ssec:strat-stab} for more details.}
These conditions, including \eqref{eq:data-rate-bnd}, are generally conservative compared to the minimal data rate required for stabilization characterized in \cite{HespanhaOrtegaVasudevan2002,TatikondaMitter2004}.
However, the gap can be made arbitrarily small by tailoring the quantization scheme to the structure of the system matrix $ A $; see, e.g., \cite[Sec.~V]{SharonLiberzon2012}.
Furthermore, the results in \cite{HespanhaOrtegaVasudevan2002,TatikondaMitter2004} assume a known bound on the disturbance, whereas the disturbance is completely unknown in \cite{SharonLiberzon2012,YangLiberzon2018} and in the present work.

\section{Main result}\label{sec:main}
The control objective is to stabilize the system defined in Section~\ref{ssec:pre-sys} under the information structure described in Section~\ref{ssec:pre-info} in a robust sense.
Specifically, we aim to establish an \emph{input-to-state stability (ISS)} \cite{Sontag1989} property with respect to the disturbance.
The theorem below adopts a characterization of ISS from \cite{SontagWang1996}; see also \cite[Sec.~II]{LiberzonNesic2007}. 
\begin{thm}\label{thm:main}
Consider the linear system \eqref{eq:lin}. Suppose that Assumptions~\ref{ass:stabilizable} and~\ref{ass:data-rate} hold.
Then there exists a communication and control strategy that yields the following ISS property with respect to the disturbance: there exist functions $ \gamma_1, \gamma_2, \gamma_3 \in \KInf $ such that, for any initial condition $ x_0 \in \R^{n_x} $ and any disturbance $ d(\cdot) $, the closed-loop solution satisfies
\begin{equation}\label{eq:iss-gs}
	|x(t)| \leq \gamma_1(|x_0|) + \gamma_2(\|d\|_{[0, \infty)}) \qquad \forall t \geq 0
\end{equation}
and
\begin{equation}\label{eq:iss-ag}
	\limsup_{t \to \infty} |x(t)| \leq \gamma_3 \left( \limsup_{t \to \infty} |d(t)| \right).
\end{equation}
\end{thm}

As in \cite{SontagWang1996}, the limit supremum in \eqref{eq:iss-ag} is understood in the essential sense of holding up to a set of measure zero.
The communication and control strategy is presented in Section~\ref{sec:strat}.
The closed-loop dynamics are derived in Sections~\ref{sec:strat} and~\ref{sec:reach}; see \eqref{eq:sys-stab} and \eqref{eq:sys-srch}.
The gain functions $ \gamma_1 $, $ \gamma_2 $, and $ \gamma_3 $ are derived in the proof of Theorem~\ref{thm:main} in Section~\ref{sec:proof}; see \eqref{eq:iss-gs-gain} and \eqref{eq:iss-ag-gain}.
In particular, $ \gamma_1 $ and $ \gamma_2 $ may depend on the choice of initial quantization parameters, but not on $ x_0 $ or $ d(\cdot) $.

\section{Communication and control strategy}\label{sec:strat}
In this section, we present the communication and control strategy used to establish Theorem~\ref{thm:main}, assuming that suitable approximations of the reachable sets of the state are available at all sampling times (their constructions are deferred to Section~\ref{sec:reach}).

The initial state $ x(0) = x_0 $ is unknown. At $ t_0 = 0 $, both the sensor and the controller are initialized with $ x^*_0 = 0 $ and an arbitrary design parameter $ E_0 > 0 $. At each sampling time $ t_k \geq 0 $, the sensor checks whether
\begin{equation}\label{eq:quant-rng}
    |x(t_k) - x^*_k| \leq E_k,
\end{equation}
that is, whether the state $ x(t_k) $ lies inside the hypercube
\begin{equation*}
    \calR_k := \{v \in \R^{n_x}: |v - x^*_k| \leq E_k\}.
\end{equation*}
The set $ \calR_k $ approximates the reachable set of the state at $ t_k $ and serves as the quantization range. If \eqref{eq:quant-rng} holds, then the state is \emph{visible} and the system is in a \emph{stabilizing stage} (Section~\ref{ssec:strat-stab}); otherwise, the state is \emph{lost} and the system is in a \emph{searching stage} (Section~\ref{ssec:strat-srch}).

If the state is visible at $ t_k $, then the system remains in a stabilizing stage until the first sampling time $ t_j > t_k $ such that $ x(t_j) \notin \calR_j $, at which point the state \emph{escapes}.
Conversely, if the state is lost at $ t_k $, then the system remains in a searching stage until the first sampling time $ t_i > t_k $ such that $ x(t_i) \in \calR_i $, at which point the state is \emph{(re)captured}.
Due to the unknown disturbance, the system may alternate between stages finitely or infinitely many times.

\subsection{Stabilizing stages}\label{ssec:strat-stab}
At each sampling time $ t_k $ in a stabilizing stage, the sensor first checks whether
\begin{equation}\label{eq:quant-cell-0}
    |x(t_k)| \leq \frac{E_k}{N}.
\end{equation}
If so, it transmits $ i_k = 1 $ to the controller. Otherwise, it partitions the range $ \calR_k $ into $ N^{n_x} $ equal hypercubic cells of radius $ E_k/N $ in the $ \infty $-norm ($ N $ per dimension), assigns each cell a unique index from $ \{ 2, \ldots, N^{n_x}+1 \}$, and transmits the index $ i_k $ of the cell containing $ x(t_k) $.

Upon receiving an $ i_k \in \{1, \ldots, N^{n_x}+1\} $, the controller infers that \eqref{eq:quant-rng} holds and reconstructs the corresponding cell center $ c_k $ using the same indexing protocol.
If $ i_k \geq 2 $, then
\begin{equation}\label{eq:quant-cell}
    |x(t_k) - c_k| \leq \frac{E_k}{N}
\end{equation}
and
\begin{equation}\label{eq:quant-cell-ctr}
    |c_k - x^*_k| \leq \frac{N - 1}{N} E_k.
\end{equation}
If $ i_k = 1 $, then $ c_k = 0 $ and \eqref{eq:quant-cell} still holds.

The controller then applies the control input
\begin{equation}\label{eq:ctrl-stab}
	u(t) = K \hat x(t), \qquad t \in [t_k, t_{k+1}),
\end{equation}
where $ K $ is the gain matrix from Assumption~\ref{ass:stabilizable}, and $ \hat x $ evolves according to the auxiliary system
\begin{equation}\label{eq:aux}
    \dot{\hat x} = A \hat x + B u 
\end{equation}
with the boundary condition
\begin{equation}\label{eq:aux-ini}
    \hat x(t_k) = c_k.
\end{equation}
That is, $ \hat x $ is updated to $ c_k $ at each sampling time $ t_k $ and is therefore only right-continuous in general.
Consequently, the controller dynamics are hybrid, comprising continuous flows and discrete jumps.

Finally, both the sensor and the controller update the quantization range by computing
\[
\begin{aligned}
    x^*_{k+1} &:= F(c_k), \\
    E_{k+1} &:= G(E_k, x_k^*)
\end{aligned}
\]
without further communication, where the functions $ F $ and $ G $ are derived in Section~\ref{ssec:reach-stab}. 

In contrast to the quantization schemes in \cite{Liberzon2003TAC,SharonLiberzon2012,Liberzon2014,YangLiberzon2018}, our scheme includes an additional quantization cell centered at the origin, corresponding to \eqref{eq:quant-cell-0}.
This design ensures that the control input vanishes whenever the state is sufficiently close to the origin, which is essential for establishing ISS; see Lemma~\ref{lem:proof-stab-kinf}.
It also preserves the equilibrium at the origin without requiring the integer $ N $ to be odd.

\subsection{Searching stages}\label{ssec:strat-srch}
At each sampling time $ t_k $ in a searching stage, the sensor transmits the ``overflow symbol'' $ i_k = 0 $.
Upon receiving $ i_k = 0 $, the controller infers that the state is lost and sets the control input to zero on $ [t_k, t_{k+1}) $.
Finally, both the sensor and the controller update the quantization range by computing
\[
\begin{aligned}
    x^*_{k+1} &:= \hatF(x^*_k), \\
    E_{k+1} &:= \hatG((1 + \varepsilon) E_k, \delta)
\end{aligned}
\]
without further communication, where the functions $ \hat F $ and $ \hat G $ are derived in Section~\ref{ssec:reach-srch}, 
and $ \varepsilon, \delta > 0 $ are arbitrary design parameters.
At an escape time $ t_j $ (i.e., the start of a searching stage), the update of $ E_{j+1} $ is modified as described in Section~\ref{ssec:reach-srch}.

\section{Approximation of reachable sets}\label{sec:reach}
In this section, we derive the update formulas used in the communication and control strategy presented in Section~\ref{sec:strat} by constructing suitable approximations of the reachable sets of the state.

\subsection{Stabilizing stages}\label{ssec:reach-stab}
Consider a sampling time $ t_k $ in a stabilizing stage.
Define the error $ e := x - \hat x $.
On $ [t_k, t_{k+1}) $, combining \eqref{eq:lin}, \eqref{eq:ctrl-stab}, and \eqref{eq:aux} yields the closed-loop system
\begin{equation}\label{eq:sys-stab}
\begin{aligned}
    \dot x &= A x + B K \hat x + D d, \\
    \dot{\hat x} &= (A + B K) \hat x,
\end{aligned}
\end{equation}
with the boundary condition \eqref{eq:aux-ini}, and thus the error dynamics
\begin{equation}\label{eq:err-stab}
	\dot e = A e + D d, \qquad |e(t_k)| = |x(t_k) - c_k| \leq \frac{E_k}{N},
\end{equation}
where the inequality follows from \eqref{eq:quant-cell}. Hence
\begin{equation}\label{eq:propag-stab-err}
\begin{aligned}
    |e(t_{k+1}^-)|
    &= \left| e^{A \tau_s} e(t_k) + \int_{t_k}^{t_{k+1}} e^{A (t_{k+1} - \tau)} D d(\tau) \dl\tau \right| \\
    &\leq \|e^{A \tau_s}\| |e(t_k)| + \bigg( \int_{0}^{\tau_s} \|e^{A s} D\| \dl s \bigg) \|d\|_{[t_k, t_{k+1}]} \\
    &\leq \frac{\Lambda}{N} E_k + \Phi \|d\|_{[t_k, t_{k+1}]},
\end{aligned}
\end{equation}
where $ \Lambda = \|e^{A \tau_s}\| $ is defined in \eqref{eq:data-rate-bnd}, and
\begin{equation}\label{eq:Phi}
	\Phi := \int_{0}^{\tau_s} \|e^{A s} D\| \dl s.
\end{equation}

We define the update formulas by
\begin{equation}\label{eq:propag-stab-x}
    x^*_{k+1} = F(c_k) := \hat x(t_{k+1}^-) = S c_k
\end{equation}
with $ S := e^{(A + B K) \tau_s} $, and
\begin{equation}\label{eq:propag-stab-E}
    E_{k+1} = G(E_k, x_k^*) := \frac{\Lambda}{N} E_k + \phi \sqrt{V_k}
\end{equation}
with
\begin{equation}\label{eq:lya-fcn}
	V_k := V(x^*_k, E_k) := (x^*_k)^\top P x^*_k + \rho E_k^2,
\end{equation}
where $ \phi, \rho > 0 $ are design parameters selected below.
The functions $ F $ and $ G $ are designed such that, if the disturbance $ \|d\|_{[t_k, t_{k+1}]} = 0 $, then $ |x(t_{k+1}) - x^*_{k+1}| \leq E_{k+1} $, so the state cannot escape at the next sampling time $ t_{k+1} $.
The term $ \phi \sqrt{V_k} $ in \eqref{eq:propag-stab-E} ensures that, if the state does escape at $ t_{k+1} $, then $ V_k $ can be bounded in terms of $ \|d\|_{[t_k,t_{k+1}]} $.

We now select three design parameters $ \psi, \rho, \phi > 0 $ sequentially so that $ V_k $ is a Lyapunov-type function with exponential decay during stabilizing stages.
Since $ A + B K $ is Hurwitz, there exist positive definite symmetric matrices $ P, Q \in \R^{n_x \times n_x} $ such that
\begin{equation}\label{eq:lya-eq}
    S^\top P S - P = -Q < 0.
\end{equation}
Define
\begin{equation}\label{eq:chi}
    \chi := \frac{2 n_x^2 \|S^\top P S\|^2}{\eigMin(Q)} + n_x \|S^\top P S\|.
\end{equation}
Since $ \Lambda < N $ by Assumption~\ref{ass:data-rate}, there exists a sufficiently small $ \psi > 0 $ such that
\begin{equation}\label{eq:psi}
    (1 + \psi) \frac{\Lambda^2}{N^2} < 1.
\end{equation}
Then there exists a sufficiently large $ \rho > 0 $ such that
\begin{equation}\label{eq:rho}
      \frac{(N - 1)^2}{N^2} \frac{\chi}{\rho} + (1 + \psi) \frac{\Lambda^2}{N^2} < 1.
\end{equation}
Finally, there exists a sufficiently small $ \phi > 0 $ such that
\begin{equation}\label{eq:nu}
\begin{aligned}
	\nu &:= \max \left\{ 1 - \frac{\eigMin(Q)}{2 \eigMax(P)},\, \frac{(N-1)^2}{N^2}\frac{\chi}{\rho} + (1 + \psi) \frac{\Lambda^2}{N^2} \right\} \\
	&\;\quad\, + \left( 1 + \frac{1}{\psi} \right) \phi^2 \rho
\end{aligned}
\end{equation}
satisfies $ \nu < 1 $.

\subsection{Searching stages}\label{ssec:reach-srch}
Consider a sampling time $ t_k $ in a searching stage.
On $ [t_k, t_{k+1}) $, combining \eqref{eq:lin} and \eqref{eq:aux} with $ u \equiv 0 $ and $ \hat x(t_k) = x_k^* $ yields the closed-loop system
\begin{equation}\label{eq:sys-srch}
\begin{aligned}
    \dot x &= A x + D d, \\
    \dot{\hat x} &= A \hat x,
\end{aligned}
\end{equation}
and thus the error dynamics
\[
	\dot e = A e + D d, \qquad e(t_k) = x(t_k) - x^*_k.
\]
Hence
\begin{equation}\label{eq:propag-srch-err}
\begin{aligned}
    |e(t_{k+1}^-)|
    &= \left| e^{A \tau_s} e(t_k) + \int_{t_k}^{t_{k+1}} e^{A (t_{k+1} - \tau)} D d(\tau) \dl\tau \right| \\
    &\leq \|e^{A \tau_s}\| |e(t_k)| + \bigg( \int_{0}^{\tau_s} \|e^{A s} D\| \dl s \bigg) \|d\|_{[t_k, t_{k+1}]} \\
    &\leq \Lambda |e(t_k)| + \Phi \|d\|_{[t_k, t_{k+1}]},
\end{aligned}
\end{equation}
where $ \Lambda $ and $ \Phi $ are defined in \eqref{eq:data-rate-bnd} and \eqref{eq:Phi}, respectively.

We define the update formulas by
\begin{equation}\label{eq:propag-srch-x}
    x^*_{k+1} = \hat F(x^*_k) := \hat x(t_{k+1}^-) = \hat S x^*_k
\end{equation}
with $ \hatS := e^{A \tau_s} $, and
\begin{equation}\label{eq:propag-srch-E}
    E_{k+1} = \hat G((1 + \varepsilon) E_k, \delta) := (1 + \varepsilon) \Lambda E_k + \Phi \delta,
\end{equation}
where $ \varepsilon, \delta > 0 $ are arbitrary design parameters introduced in Section~\ref{ssec:strat-srch}.
The functions $ \hat F $ and $ \hat G $ are designed such that
\begin{equation*}\label{eq:quant-rng-propag-srch}
    |x(t_{k+1}) - x_{k+1}^*| \le \hat G(|x(t_k)-x_k^*|, \|d\|_{[t_k,t_{k+1}]}).
\end{equation*}
Consequently, the factor $ 1 + \varepsilon $ in \eqref{eq:propag-srch-E} ensures that the growth of $ E_k $ eventually outpaces that of $|x(t_k)-x_k^*|$, thereby guaranteeing finite-time (re)capture of the state.
The parameter $ \delta $ serves as a constant disturbance estimate used only during searching stages.%
\footnote{In contrast, \cite{YangLiberzon2018} uses an increasing sequence of constants to estimate the disturbance during both stabilizing and searching stages, yielding only practical ISS. This distinction is one of the key ingredients for establishing ISS in the present paper.}

At an escape time $ t_j $ (i.e., the start of a searching stage), the update of $ E_{j+1} $ is modified. Instead of applying \eqref{eq:propag-srch-E} with $ E_j = G(E_{j-1}, x_{j-1}^*) $, we define
\begin{equation}\label{eq:propag-srch-E-esc}
    E_{j+1} := \hat G((1 + \varepsilon) \hat E_j, \delta)
\end{equation}
with
\[
	\hat E_j := \frac{\Lambda}{N} E_{j-1} + \Phi \delta.
\]
This modification ensures that the state is recaptured at $ t_{j+1} $ whenever $ \|d\|_{[t_{j-1},t_{j+1}]} \leq \delta $, which is later used in the proof of Lemma~\ref{lem:proof-cap} to bound the recapture time in terms of the disturbance.

\section{Stability analysis}\label{sec:proof}
In this section, we establish Theorem~\ref{thm:main}.
We first summarize several key intermediate results as technical lemmas in Sections~\ref{ssec:proof-stab} and~\ref{ssec:proof-srch};
their proofs are given in Appendix~\ref{appx:lem}.
Section~\ref{ssec:proof-main} then completes the proof of Theorem~\ref{thm:main}.

Throughout the analysis, we assume that $ \|d\|_{[0, \infty)} $ is finite, since otherwise the ISS bounds \eqref{eq:iss-gs} and \eqref{eq:iss-ag} hold trivially.
We also focus on the case $ \Lambda > 1 $, which holds when $ A $ has at least one eigenvalue with positive real part.%
\footnote{Note that $ \Lambda = \|e^{A \tau_s}\| > 0 $, and $ \Lambda \leq 1 $ only if all eigenvalues of $ A $ have nonpositive real parts. 
When $ \Lambda \leq 1 $, the analysis can be carried out using $ \max\{\|e^{A \tau_s}\|, 1 + \epsilon_\Lambda\} $ in place of $ \Lambda $, where $ \epsilon_\Lambda \in (0, N-1) $ can be arbitrarily small.}

\subsection{Stabilizing stages}\label{ssec:proof-stab}
We first establish that, during a stabilizing stage, the function $ V_k := V(x^*_k, E_k) $ defined in \eqref{eq:lya-fcn} is a Lyapunov-type function with exponential decay.
\begin{lem}\label{lem:proof-stab-lya}
For any sampling time $ t_k $ in a stabilizing stage,
\begin{equation}\label{eq:proof-stab-lya}
    V_{k+1} \leq \nu V_k,
\end{equation}
where $ \nu \in (0, 1) $ is defined in \eqref{eq:nu}.
\end{lem}
\begin{proof}
See Appendix~\ref{appx:proof-stab-lya}.
\end{proof}

We now relate $V_k$ to $ |x(t_k)| $, $ |x(t_{k+1})| $, and $ E_k $.
\begin{lem}\label{lem:proof-stab-lya-bnd-x}
There exist constants $ C_1, C_2, C_3 > 0 $ such that, for any sampling time $ t_k $ in a stabilizing stage,
\begin{align}
    \sqrt{V_k} &\leq C_1 (|x(t_k)| + E_k), \label{eq:proof-stab-lya-bnd-x}\\
    |x(t_k)| &\leq C_2 \sqrt{V_k}, \label{eq:proof-stab-x-bnd-lya}\\
	|x(t_{k+1})| &\leq C_3 \sqrt{V_k} + \Phi \|d\|_{[t_k, t_{k+1}]}, \label{eq:proof-stab-x-bnd-lya-propag}
\end{align}
\end{lem}
\begin{proof}
See Appendix~\ref{appx:proof-stab-lya-bnd-x}.
\end{proof}

The next two lemmas provide an exponentially decaying state bound and a state bound with suitable $ \KInf $ properties.
\begin{lem}\label{lem:proof-stab-exp}
There exist constants $ C, \lambda > 0 $ such that, for any two sampling times $ t_k > t_l $ from the same stabilizing stage,
\begin{equation}\label{eq:proof-stab-exp}
	|x(t_k)| \leq C e^{-\lambda (k - l)} (|x(t_l)| + E_l) + \Phi \|d\|_{[t_{k-1}, t_k]}.
\end{equation}
\end{lem}
\begin{proof}
See Appendix~\ref{appx:proof-stab-exp}.
\end{proof}

\begin{lem}\label{lem:proof-stab-kinf}
There exist continuous functions $ \chi^x, \chi^d: \R_{> 0} \times \R_{\geq 0} \to \R_{\geq 0} $ such that,
\begin{enumerate*}[a)]
\item 
for each fixed $ E > 0 $, $ \chi^x(E, \cdot), \chi^d(E, \cdot) \in \KInf $, 
\item 
for each fixed $ s > 0 $, $ \chi^x(\cdot, s) $ and $ \chi^d(\cdot, s) $ are nondecreasing, 
and
\item 
for any two sampling times $ t_k \geq t_l $ from the same stabilizing stage,
\end{enumerate*}
\begin{equation}\label{eq:proof-stab-kinf}
	|x(t_k)| \leq \chi^x(E_l, |x(t_l)|) + \chi^d(E_l, \|d\|_{[t_l, t_k]}).
\end{equation}
\end{lem}
\begin{proof}
See Appendix~\ref{appx:proof-stab-kinf}.
\end{proof}

Finally, we bound the state and the quantization radius at escape times in terms of the disturbance over the preceding sampling interval, independently of the initial state.
\begin{lem}\label{lem:proof-esc}
There exists a constant $ \Gamma > 0 $ such that, for any sampling time $ t_j $ at which the state escapes,
\begin{equation}\label{eq:proof-esc}
	|x(t_j)| \leq \Gamma \|d\|_{[t_{j-1}, t_j]}, \qquad E_{j-1} \leq \Gamma \|d\|_{[t_{j-1}, t_j]}.
\end{equation}
\end{lem}
\begin{proof}
See Appendix~\ref{appx:proof-esc}.
\end{proof}

\subsection{Searching stages}\label{ssec:proof-srch}
We first establish that if the state is lost at $ t_0 = 0 $, then it is guaranteed to be captured in finite time.
\begin{lem}\label{lem:proof-cap-ini}
If the state is lost at $ t_0 = 0 $, then it is captured at some sampling time $ t_{i_0} $ satisfying
\begin{equation}\label{eq:proof-cap-ini}
	i_0 
	\leq \max \left\{ 
		\eta_x \left( |x_0|/E_0 \right),\, 
		\eta_d ( \|d\|_{[0, t_{i_0}]}/\delta ) 
	\right\}
\end{equation}
with
\begin{align*}
    \eta_x(s) &:= \begin{cases}
        \lceil \log_{1 + \varepsilon} s \rceil, &s > 1, \\
        0, &0 \leq s \leq 1,
    \end{cases} \\
    \eta_d(s) &:= \begin{cases}
        \lceil \log_{1 + \varepsilon}(r_\varepsilon s) \rceil, &s > 1, \\
        0, &0 \leq s \leq 1,
    \end{cases}
\end{align*}
where
\[
	r_\varepsilon := \frac{\hat\Lambda - 1}{\Lambda - 1}, \qquad \hat\Lambda := (1 + \varepsilon) \Lambda.
\]
\end{lem}
\begin{proof}
See Appendix~\ref{appx:proof-cap-ini}.
\end{proof}

The next lemma provides bounds on the state before the first capture and on the quantization radius at the first capture.
\begin{lem}\label{lem:proof-cap-ini-kinf}
Suppose that the state is first captured at $ t_{i_0} $.
There exist continuous functions $ \hat\gamma^x_0, \hat\gamma^d_0: \R_{> 0} \times \R_{\geq 0} \to \R_{\geq 0} $ such that,
\begin{enumerate*}[a)]
\item 
for each fixed $ E > 0 $, $ \hat\gamma^x_0(E, \cdot), \hat\gamma^d_0(E, \cdot) \in \KInf $, and
\item 
for any sampling time $ t_k \leq t_{i_0} $,
\end{enumerate*}
\begin{equation}\label{eq:proof-cap-ini-kinf-x}
	|x(t_k)| 
	\leq \hat\gamma^x_0(E_0, |x_0|)
	+ \hat\gamma^d_0(E_0, \|d\|_{[0, t_{i_0}]}).
\end{equation}
Moreover, there exists a continuous function $ \hat\chi^E_0: \R_{> 0} \times \R_{\geq 0} \times \R_{\geq 0} \to \R_{> 0} $ such that,
\begin{enumerate*}[a)]
\item 
for each fixed $ E > 0 $ and $ s \geq 0 $, $ \hat\chi^E_0(E, \cdot, s) $ and $ \hat\chi^E_0(E, s, \cdot) $ are nondecreasing, and
\item
at the first capture time $ t_{i_0} $, 
\end{enumerate*}
\begin{equation}\label{eq:proof-cap-ini-E}
	E_{i_0} \leq \hat\chi^E_0(E_0, |x_0|, \|d\|_{[0, t_{i_0}]}).
\end{equation}
\end{lem}
\begin{proof}
See Appendix~\ref{appx:proof-cap-ini-kinf}.
\end{proof}

We now establish that if the state escapes, then it is guaranteed to be recaptured in finite time.
\begin{lem}\label{lem:proof-cap}
If the state escapes at $ t_j > 0 $, then it is recaptured at some sampling time $ t_i > t_j $ satisfying
\begin{equation}\label{eq:proof-cap}
	i 
	\leq j 
	+ \max \left\{ 
		\eta_d ( \|d\|_{[t_{j-1}, t_{i}]}/\delta ),\, 
		1 
	\right\},
\end{equation}
where $ \eta_d $ is defined in Lemma~\ref{lem:proof-cap-ini}.
\end{lem}
\begin{proof}
See Appendix~\ref{appx:proof-cap}.
\end{proof}

The final lemma provides bounds on the state before recapture and on the quantization radius at recapture.
\begin{lem}\label{lem:proof-cap-kinf}
Suppose that the state escapes at $ t_j $ and is recaptured at $ t_i $.
There exist $ \hat\gamma^x, \hat\gamma^d \in \KInf $ such that, for any sampling time $ t_k $ between $ t_j $ and $ t_i $,
\begin{equation}\label{eq:proof-cap-kinf-x}
	|x(t_k)| \leq \hat\gamma^x(|x(t_j)|) + \hat\gamma^d(\|d\|_{[t_{j-1}, t_{i}]}).
\end{equation}
Moreover, there exists a continuous function $ \hat\chi^E: \R_{> 0} \times \R_{\geq 0} \to \R_{> 0} $ such that
\begin{enumerate*}[a)]
\item 
for each fixed $ E > 0 $ and $ s \geq 0 $, $ \hat\chi^E(E, \cdot) $ and $ \hat\chi^E(\cdot, s) $ are nondecreasing, 
and 
\item
at the recapture time $ t_i $, 
\end{enumerate*}
\begin{equation}\label{eq:proof-cap-E}
	E_{i} \leq \hat\chi^E(E_{j-1}, \|d\|_{[t_{j-1}, t_{i}]}).
\end{equation}
\end{lem}
\begin{proof}
See Appendix~\ref{appx:proof-cap-kinf}.
\end{proof}

\subsection{Proof of Theorem~\ref{thm:main}}\label{ssec:proof-main}
The proof proceeds by alternating between searching and stabilizing stages.
During searching stages, Lemmas~\ref{lem:proof-cap-ini} and~\ref{lem:proof-cap} ensure finite-time (re)capture, while Lemmas~\ref{lem:proof-cap-ini-kinf} and~\ref{lem:proof-cap-kinf} provide bounds on the state and the quantization radius.
During stabilizing stages, Lemmas~\ref{lem:proof-stab-exp} and~\ref{lem:proof-stab-kinf} provide state bounds, while Lemma~\ref{lem:proof-esc} bounds the state and the quantization radius at escape times in terms of the disturbance.
Combining these bounds yields the ISS bounds \eqref{eq:iss-gs} and \eqref{eq:iss-ag}.

We index the alternating searching and stabilizing stages as follows.
Let $ 0 \leq i_0 < j_1 < i_1 < \cdots $ be such that the state is first captured at $ t_{i_0} $, and for each $ l \geq 1 $, it escapes at $ t_{j_l} $ and is recaptured at $ t_{i_l} $.
By Lemma~\ref{lem:proof-cap-ini}, we have $ i_0 < \infty $.
If the state never escapes after some $ t_{i_l} $, we set $ j_{l+1} = \infty $.
By Lemma~\ref{lem:proof-cap}, if $ j_l < \infty $, then $ i_l < \infty $.

First, we establish the bound \eqref{eq:iss-gs} at sampling times by constructing the functions $ \gamma_1, \gamma_2 \in \KInf $.

\paragraph*{First searching stage $ [0, t_{i_0}) $:}
By Lemma~\ref{lem:proof-cap-ini-kinf}, the bound \eqref{eq:proof-cap-ini-kinf-x} holds for all $ t_k \leq t_{i_0} $, including the case $ i_0 = 0 $ (i.e., when the state is visible at $ t_0 = 0 $).

\paragraph*{First stabilizing stage $ [t_{i_0}, t_{j_1}) $:}
At the first capture time $ t_{i_0} $, Lemma~\ref{lem:proof-cap-ini-kinf} yields
\begin{equation*}
	|x(t_{i_0})| \leq \hat\gamma^x_0(E_0, |x_0|) + \hat\gamma^d_0(E_0, \|d\|_{[0, t_{i_0}]}),
\end{equation*}
and
\begin{equation*}
	E_{i_0} \leq \hat\chi^E_0(E_0, |x_0|, \|d\|_{[0, t_{i_0}]}).
\end{equation*}
Then by Lemma~\ref{lem:proof-stab-kinf}, for any $ t_{i_0} \leq t_k < t_{j_1} $,
\begin{align*}
	|x(t_k)|
	&\leq 
	\chi^x(E_{i_0}, |x(t_{i_0})|) 
	+ \chi^d(E_{i_0}, \|d\|_{[t_{i_0}, t_k]}) \\
	&\leq 
	\chi^x \left( 
		\hat\chi^E_0(E_0, |x_0|, \|d\|_{[0, t_{i_0}]}), 
		\hat\gamma^x_0(E_0, |x_0|) + \hat\gamma^d_0(E_0, \|d\|_{[0, t_{i_0}]}) 
	\right) \\
	&\quad\, 
	+ \chi^d \left( 
		\hat\chi^E_0(E_0, |x_0|, \|d\|_{[0, t_{i_0}]}), 
		\|d\|_{[t_{i_0}, t_k]} 
	\right).
\end{align*}
Therefore, applying a standard comparison-function argument yields continuous functions $ \bar\chi_1, \bar\chi_2: \R_{> 0} \times \R_{\geq 0} \to \R_{\geq 0} $ such that
\begin{enumerate*}[a)]
\item	
for each fixed $ E > 0 $, $ \bar\chi_1(E, \cdot), \bar\chi_2(E, \cdot) \in \KInf $, and
\item
for any $ t_{i_0} \leq t_k < t_{j_1} $,
\end{enumerate*}
\[
	|x(t_k)| \leq \bar\chi_1(E_0, |x_0|) + \bar\chi_2(E_0, \|d\|_{[0, t_k]}).
\]

If $ t_{j_1} = \infty $, then the proof of \eqref{eq:iss-gs} at sampling times is complete.
Otherwise, consider an arbitrary escape time $ t_{j_l} < \infty $.

\paragraph*{Searching stage $ [t_{j_l}, t_{i_l}) $:}
At the escape time $ t_{j_l} $, Lemma~\ref{lem:proof-esc} yields
\begin{equation*}
	|x(t_{j_l})| \leq \Gamma \|d\|_{[t_{{j_l}-1}, t_{j_l}]}, \qquad E_{j_l-1} \leq \Gamma \|d\|_{[t_{{j_l}-1}, t_{j_l}]}.
\end{equation*}
Then by Lemma~\ref{lem:proof-cap-kinf}, for any $ t_{j_l} \leq t_k \leq t_{i_l} $,
\begin{align*}
	|x(t_k)| 
	&\leq \hat\gamma^x(|x(t_{j_l})|) + \hat\gamma^d(\|d\|_{[t_{{j_l}-1}, t_{i_l}]}) \\
	&\leq \hat\gamma^x(\Gamma \|d\|_{[t_{{j_l}-1}, t_{j_l}]}) + \hat\gamma^d(\|d\|_{[t_{{j_l}-1}, t_{i_l}]}) \\
	&\leq \hat\gamma(\|d\|_{[t_{{j_l}-1}, t_{i_l}]}),
\end{align*}
where $ \hat\gamma(s) := \hat\gamma^x(\Gamma s) + \hat\gamma^d(s) \in \KInf $.

\paragraph*{Stabilizing stage $ [t_{i_l}, t_{j_{l+1}}) $:}
At the recapture time $ t_{i_l} $,
\[
	|x(t_{i_l})| 
	\leq \hat\gamma(\|d\|_{[t_{{j_l}-1}, t_{i_l}]}),
\]
and Lemma~\ref{lem:proof-cap-kinf} also yields
\begin{align*}
	E_{i_l} 
	&\leq \hat\chi^E(E_{{j_l}-1}, \|d\|_{[t_{{j_l}-1}, t_{i_l}]}) \\
	&\leq \hat\chi^E(\Gamma \|d\|_{[t_{{j_l}-1}, t_{j_l}]}, \|d\|_{[t_{{j_l}-1}, t_{i_l}]}).
\end{align*}
Then by Lemma~\ref{lem:proof-stab-kinf}, for any $ t_{i_l} \leq t_k < t_{j_{l+1}} $,
\begin{align*}
	|x(t_k)|
	&\leq \chi^x(E_{i_l}, |x(t_{i_l})|) + \chi^d(E_{i_l}, \|d\|_{[t_{i_l}, t_k]}) \\
	&\leq \chi^x \left( \hat\chi^E(\Gamma \|d\|_{[t_{{j_l}-1}, t_{j_l}]}, \|d\|_{[t_{{j_l}-1}, t_{i_l}]}), \hat\gamma(\|d\|_{[t_{{j_l}-1}, t_{i_l}]}) \right) \\
	&\quad\, + \chi^d \left( \hat\chi^E(\Gamma \|d\|_{[t_{{j_l}-1}, t_{j_l}]}, \|d\|_{[t_{{j_l}-1}, t_{i_l}]}), \|d\|_{[t_{i_l}, t_k]} \right) \\
	&\leq \bar\gamma(\|d\|_{[t_{{j_l}-1}, t_{k}]}),
\end{align*}
where $ \bar\gamma(s) := \chi^x(\hat\chi^E(\Gamma s, s), \hat\gamma(s)) + \chi^d(\hat\chi^E(\Gamma s, s), s) \in \KInf $.

Combining the above cases yields the bound \eqref{eq:iss-gs} for all sampling times with the $ \KInf $ functions
\begin{equation}\label{eq:iss-gs-gain}
\begin{aligned}
	\gamma_1(s) &:= \max\{\hat\gamma^x_0(E_0, s),\, \bar\chi_1(E_0, s)\}, \\
	\gamma_2(s) &:= \max\{\hat\gamma^d_0(E_0, s),\, \bar\chi_2(E_0, s),\, \hat\gamma(s),\, \bar\gamma(s)\}.
\end{aligned}
\end{equation}

Next, we establish the bound \eqref{eq:iss-ag} at sampling times by constructing the function $ \gamma_3 \in \KInf $.
If only finitely many searching stages occur, then the system eventually remains in a stabilizing stage, and Lemma~\ref{lem:proof-stab-exp} implies
\[
	\limsup_{k \to \infty} |x(t_k)| 
	\leq \Phi \limsup_{k \to \infty} \|d\|_{[t_{k-1}, t_k]}.
\]
If infinitely many searching stages occur, then the bounds constructed in the above derivation of \eqref{eq:iss-gs} imply that, between escape times $ t_{j_l} $ and $ t_{j_{l+1}} $,
\begin{align*}
	|x(t_k)| \leq \hat\gamma(\|d\|_{[t_{{j_l}-1}, t_{i_l}]}) &\qquad \forall j_l \leq k \leq i_l, \\
	|x(t_k)| \leq \bar\gamma(\|d\|_{[t_{{j_l}-1}, t_{k}]}) &\qquad \forall i_l \leq k < j_{l+1}.
\end{align*}
Therefore,
\[
	\limsup_{k \to \infty} |x(t_k)| \leq \gamma_3 \left( \limsup_{t \to \infty} |d(t)| \right)
\]
with the $ \KInf $ function
\begin{equation}\label{eq:iss-ag-gain}
	\gamma_3(s) := \max\{\Phi s,\, \hat\gamma(s),\, \bar\gamma(s)\},
\end{equation}
where the limit supremum is understood in the essential sense.

The extension of \eqref{eq:iss-gs} and \eqref{eq:iss-ag} from sampling times to all $ t \geq 0 $ follows from standard arguments.
Specifically, for any $ t \geq 0 $, let $ k $ be such that $ t \in [t_k, t_{k+1}) $.
We first consider the case that the system is in a stabilizing stage at $ t_k $ and $ c_k \neq 0 $.
From \eqref{eq:sys-stab} and \eqref{eq:err-stab},
\begin{align*}
	|x(t)| &\leq |\hat x(t)| + |e(t)| \\
	&\leq \|e^{(A + B K) (t - t_k)}\| |c_k| \\
	&\quad\, + \|e^{A (t - t_k)}\| |e(t_k)| + \bigg( \int_{0}^{t} \|e^{A s} D\| \dl s \bigg) \|d\|_{[t_k, t]} \\
	&\leq \bar\Lambda |c_k| + \frac{\tilde\Lambda}{N} E_k + \Phi \|d\|_{[t_k, t]},
\end{align*}
where
\[
	\bar\Lambda := \max_{0 \leq s \leq \tau_s} \|e^{(A + B K) s}\|, \qquad
	\tilde\Lambda := \max_{0 \leq s \leq \tau_s} \|e^{A s}\|.
\]
By \eqref{eq:quant-cell},
\[
	|c_k| \leq |x(t_k)| + |x(t_k) - c_k| \leq |x(t_k)| + \frac{E_k}{N},
\]
and since \eqref{eq:quant-cell-0} fails at $ t_k $,
\[ |x(t_k)| > \frac{E_k}{N}. \]
Hence
\begin{equation}\label{eq:proof-inter-x}
	|x(t)| \leq \tilde H |x(t_k)| + \Phi \|d\|_{[t_k, t]},
\end{equation}
where
\[
	\tilde H:= 2 \bar\Lambda + \tilde\Lambda.
\]
If the system is in a stabilizing stage at $ t_k $ and $ c_k = 0 $, or the system is in a searching stage at $ t_k $, then $ u \equiv 0 $ on $ [t_k, t_{k+1})$, and
\begin{align*}
    |x(t)| 
    &\leq \|e^{A (t - t_k)}\| |x(t_k)| + \bigg( \int_{0}^{\tau_s} \|e^{A s} D\| \dl s \bigg) \|d\|_{[t_k, t]} \\
    &\leq \tilde\Lambda |x(t_k)| + \Phi \|d\|_{[t_k, t]},
\end{align*}
so \eqref{eq:proof-inter-x} still holds.

Since the bound \eqref{eq:iss-gs} holds for all sampling times with $ \gamma_1, \gamma_2 \in \KInf $ defined in \eqref{eq:iss-gs-gain}, by \eqref{eq:proof-inter-x}, it holds for all $ t \geq 0 $ with the $ \KInf $ functions
\begin{equation*}
\begin{aligned}
	\gamma_1(s) &:= \tilde H \max\{\hat\gamma^x_0(E_0, s),\, \bar\chi_1(E_0, s)\}, \\
	\gamma_2(s) &:= \tilde H \max\{\hat\gamma^d_0(E_0, s),\, \bar\chi_2(E_0, s),\, \hat\gamma(s),\, \bar\gamma(s)\} + \Phi s.
\end{aligned}
\end{equation*}

Since
\[
	\limsup_{k \to \infty} |x(t_k)| \leq \gamma_3 \left( \limsup_{t \to \infty} |d(t)| \right)
\]
with $ \gamma_3 \in \KInf $ defined in \eqref{eq:iss-ag-gain}, by \eqref{eq:proof-inter-x},
\begin{align*}
	\limsup_{t \to \infty} |x(t)| &\leq \limsup_{k \to \infty} \tilde H |x(t_k)| + \Phi \|d\|_{[t_k, t_{k+1}]} \\
	&\leq \tilde H \limsup_{k \to \infty} |x(t_k)| + \Phi \limsup_{t \to \infty} |d(t)|,
\end{align*}
where the limit supremum is understood in the essential sense.
Hence the bound \eqref{eq:iss-ag} holds with the $ \KInf $ function
\begin{equation*}
	\gamma_3(s) := \tilde H \max\{\Phi s,\, \hat\gamma(s),\, \bar\gamma(s)\} + \Phi s.
\end{equation*}

\section{Simulation example}\label{sec:sim}
We illustrate the proposed communication and control strategy on the linear system \eqref{eq:lin} with
\[
	A = \begin{bmatrix} 1 & 0 \\ 0 & -1.5 \end{bmatrix}, \;
	B = \begin{bmatrix} 1 \\ 0.5 \end{bmatrix}, \;
	D = \begin{bmatrix} 1 \\ 0 \end{bmatrix}, \;
	K = \begin{bmatrix} -3.5 & 0 \end{bmatrix},
\]
sampling period $ \tau_s = 0.1 $~s, and integer $ N = 5 $, which satisfy Assumptions~\ref{ass:stabilizable} and~\ref{ass:data-rate}.
The design parameters are selected as described in Sections~\ref{sec:strat} and~\ref{sec:reach}:
\[
	E_0 = 0.5, \; \varepsilon = 0.2, \; \delta = 0.1, \; \psi = 0.5, \; \rho = 150, \; \phi = 0.01.
\]

The simulation starts from the initial state
$ x_0 = \begin{bmatrix} 1 & 1 \end{bmatrix}^\top $,
which lies outside the initial quantization range.
To induce escape and recapture events, disturbance pulses of magnitude $ 1.5 $ and duration $ 0.2 $~s are applied at $ t = 3 $~s, $ 9.5 $~s, and $ 11.2 $~s; otherwise, $ d(t) = 0 $.

\begin{figure}[!t]
\centering
\begin{subfigure}{.475\textwidth}
\centering
\includegraphics[width=1\textwidth,trim={120pt 40pt 120pt 40pt},clip]{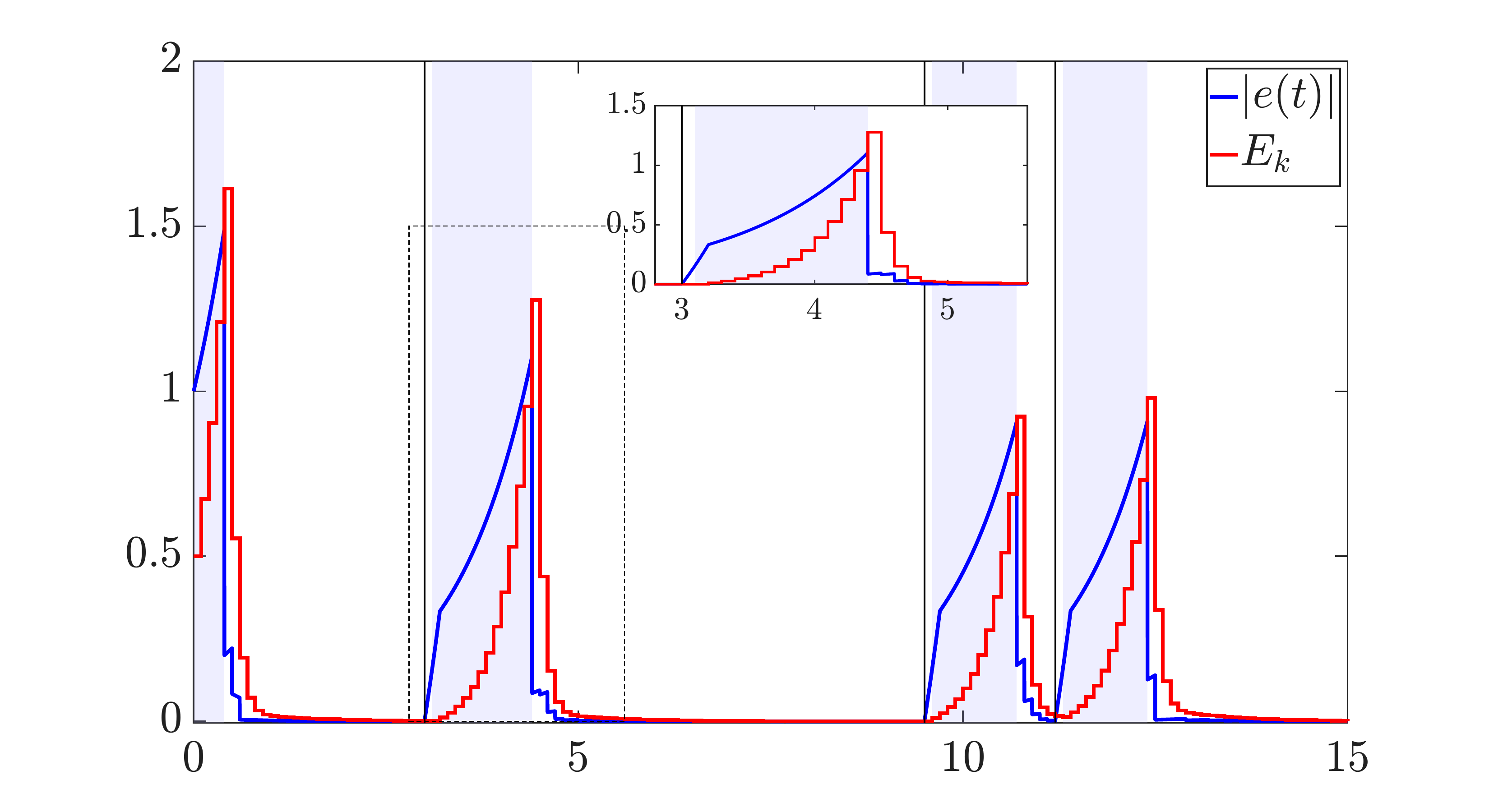}
\caption{}
\label{fig:sim-err-E}
\end{subfigure}%
\hfill
\begin{subfigure}{.475\textwidth}
\centering
\includegraphics[width=1\textwidth,trim={120pt 40pt 120pt 0pt},clip]{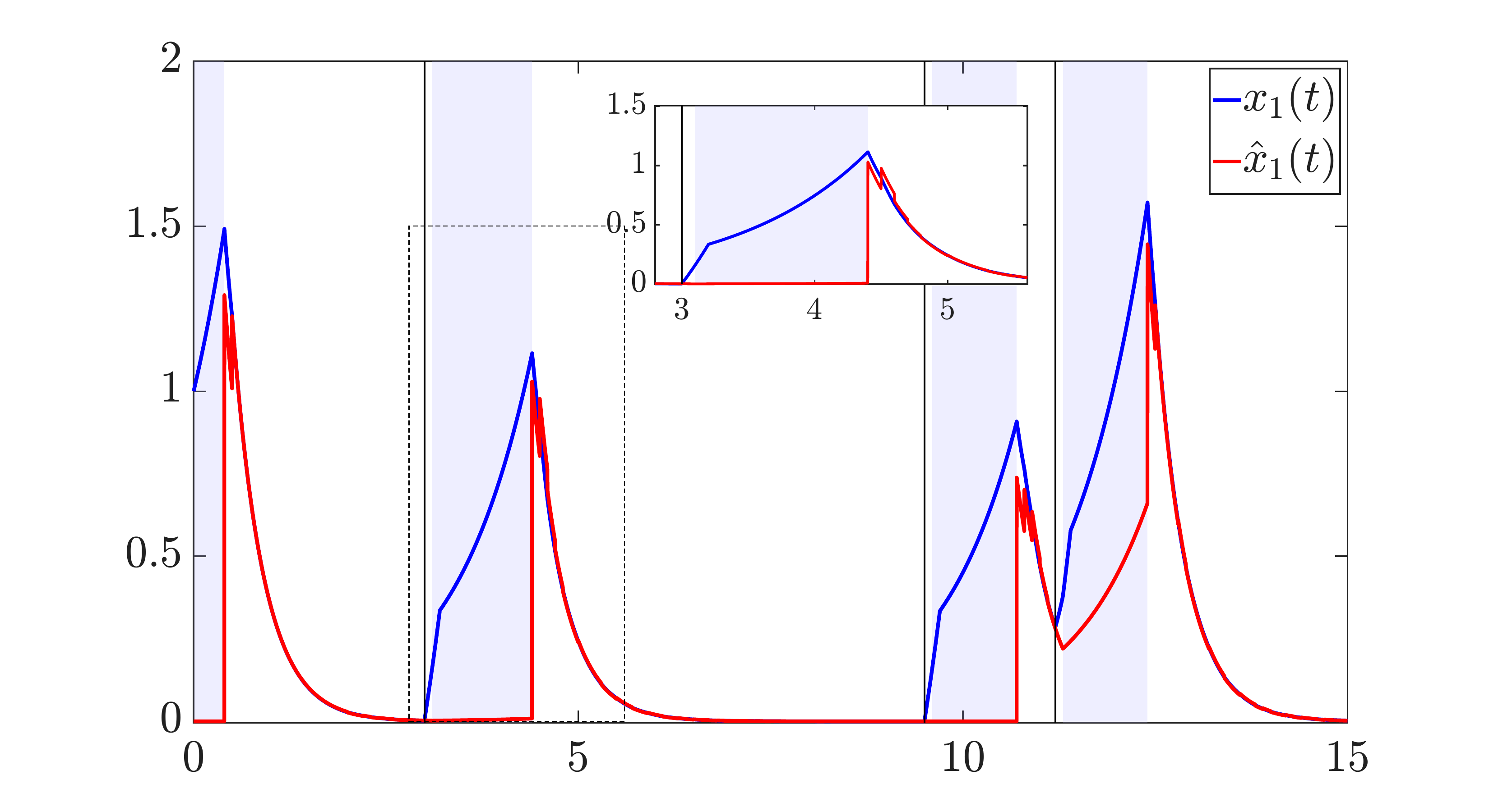}
\caption{}
\label{fig:sim-x1-aux1}
\end{subfigure}%
\caption{Evolution of
(\subref{fig:sim-err-E}) the error $ |e(t)| = |x(t) - \hat x(t)| $ and the quantization radius $ E_k $, and
(\subref{fig:sim-x1-aux1}) the first components $ x_1(t) $ and $ \hat x_1(t) $ of the system and auxiliary states.
Black vertical lines indicate disturbance onsets.
Shaded regions indicate searching stages.
The zoomed-in plots highlight the first escape and the subsequent recapture.}
\label{fig:sim}
\end{figure}
Figure~\ref{fig:sim}(\subref{fig:sim-err-E}) illustrates the escape and recapture mechanism.
During each searching stage, $ E_k $ eventually grows faster than $ |e(t_k)| $, ensuring finite-time (re)capture.
After (re)capture, both $ |e(t_k)| $ and $ E_k $ decrease rapidly at sampling times during the stabilizing stage.
Between sampling times, $ |e(t)| $ may increase due to error propagation, whereas $ E_k $, which is updated only at sampling times, remains constant and is therefore plotted as a step function.
Disturbance pulses cause subsequent escapes and initiate new searching stages.

Figure~\ref{fig:sim}(\subref{fig:sim-x1-aux1}) illustrates the corresponding evolution of the first components of the system and auxiliary states.
During searching stages, the control input is set to zero, and $ x_1(t) $ and $ \hat x_1(t) $ grow in the absence of stabilizing control.
After (re)capture, the control input $ u(t) = -3.5 \hat x_1(t) $ drives them toward zero.
The auxiliary state evolves continuously between sampling times and jumps when new quantized measurements are received.
These behaviors are consistent with the ISS properties established in Theorem~\ref{thm:main}.%

\section{Conclusion}\label{sec:end}
This paper studied feedback stabilization of linear systems under data-rate constraints in the presence of completely unknown disturbances.
Using dynamic quantization and a new disturbance-estimation design, the proposed approach establishes input-to-state stability (ISS) with respect to the disturbance.

Future work will focus on extensions to switched and nonlinear dynamics and applications to networked control systems, where communication limitations and robustness requirements play a fundamental role.

\appendix
\section{Proofs of technical lemmas}\label{appx:lem}
We first recall several basic facts from linear algebra used throughout the proofs. For any $ v, w \in \R^n $, the $ \infty $-norm satisfies
\begin{equation}\label{eq:infty-norm-vec-prod}
    |v|^2 \leq v^\top v, \quad v^\top w \leq n |v| |w|.
\end{equation}
For any symmetric matrix $ M \in \R^{n \times n} $ and any $ v \in \R^n $, 
\begin{equation}\label{eq:mat-minmax}
    \eigMin(M)\, v^\top v \leq v^\top M v \leq \eigMax(M)\, v^\top v,
\end{equation}
which yields
\begin{equation}\label{eq:mat-infty-norm}
    \eigMin(M) |v|^2 \leq v^\top M v \leq n \eigMax(M) |v|^2.
\end{equation}

\subsection{Proof of Lemma~\ref{lem:proof-stab-lya}}\label{appx:proof-stab-lya}
We first consider the case $ c_k \neq 0 $. From the update formula \eqref{eq:propag-stab-x},
\[
    x^*_{k+1} = S c_k = S (x^*_k + \Delta_k), \qquad \Delta_{k} := c_{k} - x^*_{k},
\]
where
\[ |\Delta_{k}| \leq \frac{N-1}{N} E_k \]
by \eqref{eq:quant-cell-ctr}. Then
\begin{align*}
    &\quad\,\, (x^*_{k+1})^\top P x^*_{k+1} \\
    &= (x^*_k + \Delta_k)^\top S^\top P S (x^*_k + \Delta_k) \\
    &= (x^*_k)^\top S^\top P S x^*_k + 2 (x^*_k)^\top S^\top P S \Delta_k + \Delta_k^\top S^\top P S \Delta_k.
\end{align*}
Using the Lyapunov equation \eqref{eq:lya-eq} and bounds \eqref{eq:mat-minmax} and \eqref{eq:mat-infty-norm}, we obtain
\begin{align*}
    (x^*_k)^\top S^\top P S x^*_k &= (x^*_k)^\top (P - Q) x^*_k \\
    &= (x^*_k)^\top P x^*_k - \frac{1}{2} (x^*_k)^\top Q x^*_k - \frac{1}{2} (x^*_k)^\top Q x^*_k \\
    &\leq \left( 1 - \frac{\eigMin(Q)}{2 \eigMax(P)} \right) (x^*_k)^\top P x^*_k - \frac{1}{2} \eigMin(Q) |x^*_k|^2.
\end{align*}
Moreover, by \eqref{eq:infty-norm-vec-prod},
\begin{align*}
    2 (x^*_k)^\top S^\top P S \Delta_k &\leq 2 n_x |x^*_k| \|S^\top P S\| |\Delta_k|, \\
    \Delta_k^\top S^\top P S \Delta_k &\leq n_x \|S^\top P S\| |\Delta_k|^2.
\end{align*}
Combining the bounds and completing the square yields
\begin{align*}
    &\quad\,\, (x^*_{k+1})^\top P x^*_{k+1} \\
    &\leq \bigg( 1 - \frac{\eigMin(Q)}{2 \eigMax(P)} \bigg) (x^*_k)^\top P x^*_k - \frac{1}{2} \eigMin(Q) |x^*_k|^2 \\
    &\quad\, + 2 n_x |x^*_k| \|S^\top P S\| |\Delta_k| + n_x \|S^\top P S\| |\Delta_k|^2 \\
    &\leq \bigg( 1 - \frac{\eigMin(Q)}{2 \eigMax(P)} \bigg) (x^*_k)^\top P x^*_k + \chi |\Delta_k|^2 \\
    &\quad\, - \bigg( \sqrt{\frac{1}{2} \eigMin(Q)} |x^*_k| - \frac{\sqrt{2} n_x \|S^\top P S\|}{\sqrt{\eigMin(Q)}} |\Delta_k| \bigg)^2 \\
    &\leq \bigg( 1 - \frac{\eigMin(Q)}{2 \eigMax(P)} \bigg) (x^*_k)^\top P x^*_k + \frac{(N - 1)^2}{N^2} \chi E_k^2,
\end{align*}
where $ \chi $ is defined in \eqref{eq:chi}.

Next, from the update formula \eqref{eq:propag-stab-E} and Young's inequality,
\begin{align*}
    E_{k+1}^2
    &= \bigg( \frac{\Lambda}{N} E_k + \phi \sqrt{V_k} \bigg)^2 \\
    &\leq (1 + \psi) \frac{\Lambda^2}{N^2} E_k^2 + \left( 1 + \frac{1}{\psi} \right) \phi^2 V_k
\end{align*}
where $ \psi > 0 $ is chosen according to \eqref{eq:psi}.

Finally, combining both bounds in $ V_{k+1} = (x^*_{k+1})^\top P x^*_{k+1} + \rho E_{k+1}^2 $ yields
\[
V_{k+1} \le \nu V_k,
\]
where $ \rho > 0 $ is chosen according to \eqref{eq:rho}, and $ \nu \in (0, 1) $ is defined in \eqref{eq:nu}.

If $ c_k = 0 $, then $ x^*_{k+1} = S c_k = 0 $, and
\begin{align*}
    V_{k+1} &= \rho E_{k+1}^2 \leq \left( (1 + \psi) \frac{\Lambda^2}{N^2} + \left( 1 + \frac{1}{\psi} \right) \phi^2 \rho \right) V_k,
\end{align*}
so \eqref{eq:proof-stab-lya} still holds.
%

\subsection{Proof of Lemma~\ref{lem:proof-stab-lya-bnd-x}}\label{appx:proof-stab-lya-bnd-x}
We first prove \eqref{eq:proof-stab-lya-bnd-x}.
Since $ V_k = (x^*_k)^\top P x^*_k + \rho E_k^2 $, the upper bound in \eqref{eq:mat-infty-norm} yields
\begin{align*}
	\sqrt{V_k} \leq \sqrt{n_x \eigMax(P)} |x^*_k| + \sqrt{\rho} E_k,
\end{align*}
where we used $ \sqrt{a + b} \leq \sqrt{a} + \sqrt{b} $. Moreover, by \eqref{eq:quant-rng},
\[
	|x^*_k| \leq |x(t_k)| + |x(t_k) - x^*_k| \leq |x(t_k)| + E_k.
\]
Hence \eqref{eq:proof-stab-lya-bnd-x} holds with
\[
	C_1 := \sqrt{n_x \eigMax(P)} + \sqrt{\rho}.
\]

Next, we establish \eqref{eq:proof-stab-x-bnd-lya}.
From \eqref{eq:quant-rng},
\begin{align*}
	|x(t_k)| &\leq |x^*_k| + |x(t_k) - x^*_k| \leq |x^*_k| + E_k.
\end{align*}
Using the lower bound in \eqref{eq:mat-infty-norm}, we obtain
\[
	|x_k^*| \leq \sqrt{\frac{V_k}{\eigMin(P)}}, \qquad E_k \leq \sqrt{\frac{V_k}{\rho}}
\]
which yields \eqref{eq:proof-stab-x-bnd-lya} with
\[
	C_2 := \frac{1}{\sqrt{\eigMin(P)}} + \frac{1}{\sqrt{\rho}}.
\]

Finally, we prove \eqref{eq:proof-stab-x-bnd-lya-propag}.
From the error estimate \eqref{eq:propag-stab-err} and update formula \eqref{eq:propag-stab-x},
\begin{align*}
	|x(t_{k+1})| &\leq |\hat x(t_{k+1}^-)| + |e(t_{k+1}^-)| \\
	&\leq \|S\| |c_k| + \frac{\Lambda}{N} E_k + \Phi \|d\|_{[t_k, t_{k+1}]}.
\end{align*}
We first consider the case $ c_k \neq 0 $. By \eqref{eq:quant-cell-ctr} and the lower bound in \eqref{eq:mat-infty-norm},
\begin{align*}
	|c_k| &\leq |x_k^*| + |c_k - x_k^*| \leq |x_k^*| + \frac{N - 1}{N} E_k.
\end{align*}
Hence
\begin{align*}
	|x(t_{k+1})|
	&\leq \|S\| |x_k^*| + \frac{(N-1) \|S\| + \Lambda}{N} E_k + \Phi \|d\|_{[t_k, t_{k+1}]}.
\end{align*}
Bounding $ |x^*_k| $ and $ E_k $ in terms of $ V_k $ yields \eqref{eq:proof-stab-x-bnd-lya-propag} with
\[
	C_3 := \frac{\|S\|}{\sqrt{\eigMin(P)}} + \frac{(N-1) \|S\| + \Lambda}{N \sqrt{\rho}}.
\]
If $ c_k = 0 $, then $ \hat x(t_{k+1}^-) = S c_k = 0 $, and
\[
	|x(t_{k+1})| = |e(t_{k+1}^-)| \leq \frac{\Lambda}{N} E_k + \Phi \|d\|_{[t_k, t_{k+1}]},
\]
so \eqref{eq:proof-stab-x-bnd-lya-propag} still holds.

\subsection{Proof of Lemma~\ref{lem:proof-stab-exp}}\label{appx:proof-stab-exp}
By iterating Lemma~\ref{lem:proof-stab-lya}, we obtain
\[
	V_{k-1} \leq \nu^{k-1-l} V_l,
\]
where $ \nu \in (0, 1) $. Applying \eqref{eq:proof-stab-lya-bnd-x} at $ t_l $ and \eqref{eq:proof-stab-x-bnd-lya-propag} at $ t_{k-1} $ yields
\begin{align*}
	|x(t_k)| &\le C_3 \sqrt{V_{k-1}} + \Phi \|d\|_{[t_{k-1}, t_k]} \\
	&\leq C_3 \nu^{(k-l-1)/2} \sqrt{V_l} + \Phi \|d\|_{[t_{k-1}, t_k]} \\
	&\leq C_1 C_3 \nu^{(k-l-1)/2} (|x(t_l)| + E_l) + \Phi \|d\|_{[t_{k-1}, t_k]}.
\end{align*}
Hence \eqref{eq:proof-stab-exp} holds with
\[
	C := C_1 C_3 \nu^{-1/2}, \qquad \lambda := -\frac{1}{2} \ln\nu > 0.
\]
%

\subsection{Proof of Lemma~\ref{lem:proof-stab-kinf}}\label{appx:proof-stab-kinf}
We consider two cases depending on whether the initial state dominates the disturbance.

\paragraph{Case~1: $ |x(t_l)| > \Phi \|d\|_{[t_l, t_k]} $.}
Let
\[ H := 2 \|S\| + \Lambda > 1. \]
Then $ \log_\nu H < 0 $ since $ \nu \in (0, 1) $.
Fix any $ \kappa \in (0, -1/\log_\nu H) $, and define
\begin{equation*}
    l_x := \max \left\{ \left\lceil \kappa \log_\nu |x(t_l)| \right\rceil,\, 0 \right\}.
\end{equation*}
Then $ \nu^{l_x} \leq |x(t_l)|^\kappa $ for all $ |x(t_l)| > 0 $.
We distinguish between the ``long-time'' regime $ k \geq l + l_x $ and, when $ l_x > 0 $, the ``short-time'' regime $ k < l + l_x $.

If $ k \geq l + l_x $, iterating Lemma~\ref{lem:proof-stab-lya} yields
\[
	V_k \leq \nu^{k-l} V_l \leq \nu^{l_x} V_l \leq |x(t_l)|^{\kappa} V_l.
\]
Applying \eqref{eq:proof-stab-lya-bnd-x} at $ t_l $ and \eqref{eq:proof-stab-x-bnd-lya} at $ t_k $, we obtain
\begin{align*}
    |x(t_k)| &\leq C_2 \sqrt{V_k} \\
    &\leq C_2 |x(t_l)|^{\kappa/2} \sqrt{V_l} \\
    &\leq C_1 C_2 |x(t_l)|^{\kappa/2} (|x(t_l)| + E_l) =: \chi_1^x(E_l, |x(t_l)|),
\end{align*}
where 
for each fixed $ E > 0 $, $ \chi_1^x(E, \cdot) \in \KInf $, 
and
for each fixed $ s > 0 $, $ \chi_1^x(\cdot, s) $ is nondecreasing.

If $ l_x > 0 $ and $ l \leq k < l + l_x $, we first consider the case $ c_k \neq 0 $.
From the error estimate \eqref{eq:propag-stab-err} and update formula \eqref{eq:propag-stab-x},
\begin{align*}
	|x(t_{k+1})| &\leq |\hat x(t_{k+1}^-)| + |e(t_{k+1}^-)| \\
	&\leq \|S\| |c_k| + \frac{\Lambda}{N} E_k + \Phi \|d\|_{[t_k, t_{k+1}]}.
\end{align*}
By \eqref{eq:quant-cell},
\[
	|c_k| \leq |x(t_k)| + |x(t_k) - c_k| \leq |x(t_k)| + \frac{E_k}{N},
\]
and since \eqref{eq:quant-cell-0} fails at $ t_k $,
\[ |x(t_k)| > \frac{E_k}{N}. \]
Hence
\begin{equation}\label{eq:proof-stab-kinf-x-short}
\begin{aligned}
	|x(t_{k+1})| &\leq H |x(t_k)| + \Phi \|d\|_{[t_k, t_{k+1}]}.
\end{aligned}
\end{equation}
If $ c_k = 0 $, then $ u \equiv 0 $ on $ [t_k, t_{k+1})$, and
\begin{align*}
    |x(t_{k+1})| 
    &\leq \|e^{A \tau_s}\| |x(t_k)| + \bigg( \int_{0}^{\tau_s} \|e^{A s} D\| \dl s \bigg) \|d\|_{[t_k, t_{k+1}]} \\
    &\leq \Lambda |x(t_k)| + \Phi \|d\|_{[t_k, t_{k+1}]},
\end{align*}
so \eqref{eq:proof-stab-kinf-x-short} still holds.
Since \eqref{eq:proof-stab-kinf-x-short} holds for all sampling times between $ t_l $ and $ t_k $, iterating it yields
\begin{align*}
	|x(t_k)| &\leq H^{k-l} |x(t_l)| + \frac{H^{k-l} - 1}{H - 1} \Phi \|d\|_{[t_l, t_k]} \\
	&\leq \frac{H^{k-l+1} - 1}{H - 1} |x(t_l)| \\
	&\leq \frac{H^{l_x} - 1}{H - 1} |x(t_l)| =: \chi_2^x(|x(t_l)|),
\end{align*}
where $ |x(t_l)| < 1 $ because $ l_x > 0 $ in the ``short-time'' regime.
For all $ s \in (0, 1) $, we have $ l_x(s) < \kappa \log_\nu s + 1 $, and thus
\begin{align*}
	\chi_2^x(s) &< \frac{H^{\kappa \log_\nu s + 1}}{H - 1} s = \frac{H}{H - 1} s^{\kappa \log_\nu H + 1} =: \chi_3^x(s),
\end{align*}
where $ \chi_3^x \in \KInf $ by the choice of $ \kappa \in (0, -1/\log_\nu H) $.

Finally,
\begin{align*}
	|x(t_k)|
	&\leq \max\{\chi_1^x(E_l, |x(t_l)|),\, \chi_3^x(|x(t_l)|)\} \\
	&=: \chi^x(E_l, |x(t_l)|),
\end{align*}
which satisfies the stated $ \KInf $ and monotonicity properties.

\paragraph{Case~2: $ |x(t_l)| \leq \Phi \|d\|_{[t_l, t_k]} $.}
The proof follows the same steps as in Case~1 and is omitted for brevity. It yields analogous functions $ \chi_1^d $ and $ \chi_3^d $ such that
\begin{align*}
	|x(t_k)| 
	&\leq \max\{\chi_1^d(E_l, \|d\|_{[t_l, t_k]}),\, \chi_3^d(\|d\|_{[t_l, t_k]})\} \\
	&=: \chi^d(E_l, \|d\|_{[t_l, t_k]}),
\end{align*}
which satisfies the stated $ \KInf $ and monotonicity properties.

Combining the two cases yields \eqref{eq:proof-stab-kinf}.
%

\subsection{Proof of Lemma~\ref{lem:proof-esc}}\label{appx:proof-esc}
Since \eqref{eq:quant-rng} holds at $ t_{j-1} $ but fails at $ t_j $,
\[
	|e(t_j)| = |x(t_j) - x^*_j| > E_j.
\]
Applying the error estimate \eqref{eq:propag-stab-err} and update formula \eqref{eq:propag-stab-E} at $ t_{j-1} $, we obtain
\[
	\frac{\Lambda}{N} E_{j-1} + \Phi \|d\|_{[t_{j-1}, t_j]} > \frac{\Lambda}{N} E_{j-1} + \phi \sqrt{V_{j-1}},
\]
that is,
\[
	\phi \sqrt{V_{j-1}} < \Phi \|d\|_{[t_{j-1}, t_j]}.
\]
Therefore,
\begin{align*}
	E_{j-1} &\leq \sqrt{\frac{V_{j-1}}{\rho}} < \frac{1}{\phi \sqrt{\rho}} \Phi \|d\|_{[t_{j-1}, t_j]}.
\end{align*}
Next, applying \eqref{eq:proof-stab-x-bnd-lya-propag} at $ t_{j-1} $ yields
\begin{align*}
	|x(t_j)| &\leq C_3 \sqrt{V_{j-1}} + \Phi \|d\|_{[t_{j-1}, t_j]} \\
	&< \left( \frac{C_3}{\phi} + 1 \right) \Phi \|d\|_{[t_{j-1}, t_j]}.
\end{align*}
Combining the two bounds, we obtain \eqref{eq:proof-esc} with
\[
	\Gamma := \max \left\{ \frac{1}{\phi \sqrt{\rho}},\, \frac{C_3}{\phi} + 1 \right\}\, \Phi > 0,
\]
where $ C_3 > 0 $ is from Lemma~\ref{lem:proof-stab-lya-bnd-x}.
%

\subsection{Proof of Lemma~\ref{lem:proof-cap-ini}}\label{appx:proof-cap-ini}
Let $ t_k > 0 $ be a sampling time such that the state has not yet been captured by $ t_k $.
Then the system remains in a searching stage on $ [0, t_k) $, and iterating the error estimate \eqref{eq:propag-srch-err} and update formulas \eqref{eq:propag-srch-x} and \eqref{eq:propag-srch-E} from $ 0 $ to $ t_k $ with $ x^*_0 = 0 $ yields
\begin{align*}
	|e(t_k^-)| &\leq \Lambda^{k} |x_0| + \frac{\Lambda^{k} - 1}{\Lambda - 1} \Phi \|d\|_{[0, t_k]}, \\
	E_{k} &= \hat\Lambda^{k} E_0 + \frac{\hat\Lambda^{k} - 1}{\hat\Lambda - 1} \Phi \delta.
\end{align*}
Let $ i_0^* $ be the smallest positive integer such that
\[
	i_0^* 
	\geq \max \left\{ 
		\eta_x \left( |x_0|/E_0 \right),\, 
		\eta_d \left( \|d\|_{[0, t_{i_0^*}]}/\delta \right)
	\right\}.
\]
Such an integer exists if $ \|d\|_{[0, \infty)} $ is finite.
We show that the state is captured no later than $t_{i_0^*}$.

First, since $ i_0^* \geq \eta_x(|x_0|/E_0) $,
\begin{align*}
	\hat\Lambda^{i_0^*} E_0
	&= \Lambda^{i_0^*} (1 + \varepsilon)^{i_0^*} E_0 \\
	&\geq \Lambda^{i_0^*} (1 + \varepsilon)^{\eta_x(|x_0|/E_0)} E_0
	\geq \Lambda^{i_0^*} |x_0|.
\end{align*}
Second, if $ \delta < \|d\|_{[0, t_{i_0^*}]} $, then 
\[
	i_0^* 
	\geq \eta_d \left( \|d\|_{[0, t_{i_0^*}]}/\delta \right) 
	\geq \log_{1 + \varepsilon} \left( r_\varepsilon \|d\|_{[0, t_{i_0^*}]}/\delta \right),
\]
and
\begin{align*}
    \frac{\hat\Lambda^{i_0^*} - 1}{\hat\Lambda - 1} \delta &> \dfrac{\Lambda^{i_0^*} - 1}{\hat\Lambda - 1} (1 + \varepsilon)^{i_0^*} \delta \\
    &= \dfrac{\Lambda^{i_0^*} - 1}{\Lambda - 1} \frac{(1 + \varepsilon)^{i_0^*}}{r_\varepsilon}\, \delta 
    \geq \dfrac{\Lambda^{i_0^*} - 1}{\Lambda - 1} \|d\|_{[0, t_{i_0^*}]}.
\end{align*}
Otherwise, $ \delta \geq \|d\|_{[0, t_{i_0^*}]} $, and since $ \hat\Lambda > \Lambda $,
\[
	\frac{\hat\Lambda^{i_0^*} - 1}{\hat\Lambda - 1} \delta \geq \frac{\Lambda^{i_0^*} - 1}{\Lambda - 1} \|d\|_{[0, t_{i_0^*}]}.
\]

Combining these bounds yields
\[
	E_{i_0^*} \geq |e(t_{i_0^*})|. 
\]
Hence the state is captured at some sampling time $ t_{i_0} \leq t_{i_0^*} $.
If $ i_0 < i_0^* $, then \eqref{eq:proof-cap-ini} follows immediately. If $ i_0 = i_0^* $, then
\[
	i_0 - 1 
	< \max \left\{ 
		\eta_x \left( |x_0|/E_0 \right),\, 
		\eta_d \left( \|d\|_{[0, t_{i_0-1}]}/\delta \right) 
	\right\},
\]
and since both sides are integers, \eqref{eq:proof-cap-ini} still holds.
%

\subsection{Proof of Lemma~\ref{lem:proof-cap-ini-kinf}}\label{appx:proof-cap-ini-kinf}
Consider any $ t_k \le t_{i_0} $. Since the state is first captured at $ t_{i_0} $, the system remains in a searching stage on $ [0, t_{i_0}) $.

First, iterating the error estimate \eqref{eq:propag-srch-err} and update formula \eqref{eq:propag-srch-x} from $ 0 $ to $ t_k $ with $ x^*_0 = 0 $ yields
\begin{align*}
	|x(t_k)| &\leq \Lambda^{k} |x_0| + \frac{\Lambda^{k} - 1}{\Lambda - 1} \Phi \|d\|_{[0, t_k]} \\
	&\leq \Lambda^{i_0} |x_0| + \frac{\Lambda^{i_0} - 1}{\Lambda - 1} \Phi \|d\|_{[0, t_{i_0}]}.
\end{align*}
Substituting the bound on $ i_0 $ from \eqref{eq:proof-cap-ini} into the right-hand side, we obtain a bound on $ |x(t_k)| $ that, for each fixed $ E_0 > 0 $, is nondecreasing in $ |x_0| $ and $ \|d\|_{[0, t_{i_0}]} $ and vanishes when $ |x_0| = \|d\|_{[0, t_{i_0}]} = 0 $.
Therefore, applying a standard comparison-function argument yields continuous functions $ \hat\gamma^x_0, \hat\gamma^d_0: \R_{> 0} \times \R_{\geq 0} \to \R_{\geq 0} $ with the stated $ \KInf $ properties such that \eqref{eq:proof-cap-ini-kinf-x} holds.

Second, iterating the update formula \eqref{eq:propag-srch-E} from $ 0 $ to $ t_{i_0} $ yields
\begin{align*}
	E_{i_0} &= \hat\Lambda^{i_0} E_0 + \frac{\hat\Lambda^{i_0} - 1}{\hat\Lambda - 1} \Phi \delta.
\end{align*}
Substituting again the bound on $ i_0 $ from \eqref{eq:proof-cap-ini} into the right-hand side, we obtain a bound on $ E_{i_0} $ that, for each fixed $ E_0 > 0 $, is nondecreasing in $ |x_0| $ and $ \|d\|_{[0, t_{i_0}]} $.
Therefore, applying another standard comparison-function argument yields a continuous function $ \hat\chi^E_0: \R_{> 0} \times \R_{\geq 0} \times \R_{\geq 0} \to \R_{> 0} $ with the stated monotonicity property such that \eqref{eq:proof-cap-ini-E} holds.

\subsection{Proof of Lemma~\ref{lem:proof-cap}}\label{appx:proof-cap}
If the state is recaptured at $ t_{j+1} $, then $ i = j + 1 $ and \eqref{eq:proof-cap} holds.
In the remainder of the proof, we assume that the state is not recaptured at $ t_{j+1} $.

The system remains in a searching stage on $ [t_j, t_i) $.
Applying the same argument as in the proof of Lemma~\ref{lem:proof-cap-ini}, with the adjustment at $ t_j $ according to \eqref{eq:propag-srch-E-esc}, we obtain
\begin{equation}\label{eq:proof-cap-ini-esc}
	i \leq 
	j + \max \left\{ 
		\eta_x \left( |e(t_j)|/\hatE_j \right),\, 
		\eta_d \left( \|d\|_{[t_j, t_i]}/\delta \right),\, 
		1 
	\right\},
\end{equation}
where $ \eta_x, \eta_d $ are defined in Lemma~\ref{lem:proof-cap-ini} and satisfy $\eta_x(s) \leq \eta_d(s) $ for all $ s \geq 0 $.

Next, applying the searching-stage error estimate \eqref{eq:propag-srch-err} at $ t_j $ and the stabilizing-stage error estimate \eqref{eq:propag-stab-err} at $ t_{j-1} $ yields
\begin{align*}
	|e(t_{j+1})| &\leq \Lambda |e(t_j)| + \Phi \|d\|_{[t_j, t_{j+1}]} \\
	&\leq \frac{\Lambda^2}{N} |e(t_{j-1})| + (\Lambda + 1) \Phi \|d\|_{[t_{j-1}, t_{j+1}]}.
\end{align*}
Moreover, from the adjusted update formula \eqref{eq:propag-srch-E-esc},
\begin{align*}
	E_{j+1} &= \hat\Lambda \hatE_j + \Phi \delta = \frac{\hat\Lambda \Lambda}{N} E_{j-1} + (\hat\Lambda + 1) \Phi \delta,
\end{align*}
where $ \hat\Lambda = (1 + \varepsilon) \Lambda $.

Since the state escapes at $ t_j $ and is not recaptured at $ t_{j+1} $,
\[
	|e(t_{j-1})| \leq E_{j-1}, \qquad |e(t_{j+1})| > E_{j+1},
\]
and the above bound on $ |e(t_{j+1})| $ and formula of $ E_{j+1} $ imply
\[
	\delta < \|d\|_{[t_{j-1}, t_{j+1}]}.
\]
Consequently,
\[
	\frac{|e(t_{j})|}{\hat E_{j}} \leq \frac{\frac{\Lambda}{N} |e(t_{j-1})| + \Phi \|d\|_{[t_{j-1}, t_{j+1}]}}{\frac{\Lambda}{N} E_{j-1} + \Phi \delta} < \frac{\|d\|_{[t_{j-1}, t_{j+1}]}}{\delta},
\]
and thus
\begin{align*}
    \eta_x \left( |e(t_j)|/\hat E_j \right) 
    \leq \eta_x \left( \|d\|_{[t_{j-1}, t_{j+1}]}/\delta \right) 
    \leq \eta_d \left( \|d\|_{[t_{j-1}, t_{j+1}]}/\delta \right).
\end{align*}
Substituting this bound into \eqref{eq:proof-cap-ini-esc} yields \eqref{eq:proof-cap}.
%

\subsection{Proof of Lemma~\ref{lem:proof-cap-kinf}}\label{appx:proof-cap-kinf}
Since the state escapes at $ t_j $ and is recaptured at $ t_i $, the system remains in a searching stage on $ [t_j, t_i) $. Consider any sampling time $ t_k \in [t_j, t_i) $.

First, since $ u \equiv 0 $ on $ [t_j, t_i) $, the state satisfies
\[ \dot x = A x + D d, \]
and thus
\begin{align*}
	|x(t_{k+1})| &\leq \Lambda |x(t_k)| + \Phi \|d\|_{[t_k, t_{k+1}]}.
\end{align*}
Iterating this bound from $ t_j $ to $ t_k $ yields 
\begin{align*}
	|x(t_k)| &\leq \Lambda^{k-j} |x(t_j)| + \frac{\Lambda^{k-j} - 1}{\Lambda - 1} \Phi \|d\|_{[t_j, t_k]} \\
	&\leq \Lambda^{i-j} |x(t_j)| + \frac{\Lambda^{i-j} - 1}{\Lambda - 1} \Phi \|d\|_{[t_j, t_i]}.
\end{align*}
Substituting the bound on $ i $ from \eqref{eq:proof-cap} into the right-hand side, we obtain a bound on $ |x(t_k)| $ that is nondecreasing in $ |x(t_j)| $ and $ \|d\|_{[t_{j-1}, t_{i}]} $ and vanishes when $ |x(t_j)| = \|d\|_{[t_{j-1}, t_{i}]} = 0 $.
Therefore, it can be dominated by $ \KInf $ functions $ \hat\gamma^x(|x(t_j)|) + \hat\gamma^d(\|d\|_{[t_{j-1}, t_{i}]}) $ using a standard comparison-function argument, which proves \eqref{eq:proof-cap-kinf-x}.

Second, iterating the update formula \eqref{eq:propag-srch-E} from $ t_j $ to $ t_i $, with the adjustment at $ t_j $ according to \eqref{eq:propag-srch-E-esc}, yields
\begin{align*}
	E_{i} &= \hat \Lambda^{i-j} \hat E_j + \frac{\hat\Lambda^{i-j} - 1}{\hat\Lambda - 1} \Phi \delta \\
	&= \frac{\hat\Lambda^{i-j} \Lambda}{N} E_{j-1} + \frac{\hat\Lambda^{i-j+1} - 1}{\hat\Lambda - 1} \Phi \delta.
\end{align*}
Substituting again the bound on $ i $ from \eqref{eq:proof-cap} into the right-hand side, we obtain a bound on $ E_{i} $ that is nondecreasing in each of $ E_{j-1} $ and $ \|d\|_{[t_{j-1}, t_{i}]} $.
Therefore, applying another standard comparison-function argument yields a continuous function $ \hat\chi^E: \R_{> 0} \times \R_{\geq 0} \to \R_{> 0} $ with the stated monotonicity property such that \eqref{eq:proof-cap-E} holds.


\addcontentsline{toc}{section}{References}
\bibliography{reference,reference-temp}

\end{document}